\documentclass{article}

\usepackage[
	persons={Nathan, Nima, Amin, Aaron},
	bibliosources=refs.bib
]{pomegranate}

\DeclareClass{\PRAM}
\DeclareClass{\NC}
\DeclareClass{\RNC}

\DeclareOperator{\poly}
\DeclareOperator{\polylog}

\Tag{
	\DeclareTheorem<remark>{question}
}
\Tag<lipics>{
	\theoremstyle{remark}
	\newtheorem{question}[theorem]{Question}
}

\title{Sampling Arborescences in Parallel}

\Tag{
	\author{Nima Anari}
	\author{Nathan Hu}
	\author{Amin Saberi}
	\affil{Stanford University, \textsf{\{anari,zixia314,saberi\}@stanford.edu}}
	\author{Aaron Schild}
	\affil{University of Washington, \textsf{aschild@berkeley.edu}}
}

\Tag<lipics>{
	\author{Nima Anari}{Stanford University, USA}{anari@cs.stanford.edu}{}{}
	\author{Nathan Hu}{Stanford University, USA}{zixia314@stanford.edu}{}{}
	\author{Amin Saberi}{Stanford University, USA}{saberi@stanford.edu}{}{}
	\author{Aaron Schild}{University of Washington, USA}{aschild@berkeley.edu}{}{}
	
	\authorrunning{Nima Anari, Nathan Hu, Amin Saberi, and Aaron Schild}
	
	\Copyright{Nima Anari, Nathan Hu, Amin Saberi, and Aaron Schild}
	
	
}

\Tag<lipics>{
	\EventEditors{James R. Lee}
	\EventNoEds{1}
	\EventLongTitle{12th Innovations in Theoretical Computer Science Conference (ITCS 2021)}
	\EventShortTitle{ITCS 2021}
	\EventAcronym{ITCS}
	\EventYear{2021}
	\EventDate{January 6--8, 2021}
	\EventLocation{Virtual Conference}
	\EventLogo{}
	\SeriesVolume{185}
	\ArticleNo{41}
}

\begin{document}
	\maketitle
	
	\begin{abstract}
		We study the problem of sampling a uniformly random directed rooted spanning tree, also known as an arborescence, from a possibly weighted directed graph. Classically, this problem has long been known to be polynomial-time solvable; the exact number of arborescences can be computed by a determinant \cite{Tutte48}, and sampling can be reduced to counting \cite{JVV86, JS96}. However, the classic reduction from sampling to counting seems to be inherently sequential. This raises the question of designing efficient parallel algorithms for sampling. We show that sampling arborescences can be done in \RNC{}.
		
		For several well-studied combinatorial structures, counting can be reduced to the computation of a determinant, which is known to be in \NC{} \cite{Csa75}. These include arborescences, planar graph perfect matchings, Eulerian tours in digraphs, and determinantal point processes. However, not much is known about efficient parallel sampling of these structures. Our work is a step towards resolving this mystery.
	\end{abstract}

	\section{Introduction}
\label{sec:intro}

Algorithms for (approximately) counting various combinatorial structures are often based on the equivalence between (approximate) counting and sampling \cite{JVV86, JS96}. This is indeed the basis of the Markov Chain Monte Carlo (MCMC) method to approximate counting, which is arguably the most successful approach to counting, resolving long-standing problems such as approximating the permanent \cite{JSV04} and computing the volume of convex sets \cite{DFK91}.

Approximate sampling and counting are known to be equivalent for a wide class of problems, including the so-called self-reducible ones \cite{JVV86, JS96}. This equivalence is nontrivial and most useful in the direction of reducing counting to sampling. However, for some problems, the ``easier'' direction of this equivalence, namely the reduction from sampling to counting, proves useful. For these problems, almost by definition, we can count via approaches other than MCMC.

One of the mysterious approaches to counting is via determinant computations. A range of counting problems can be (exactly) solved by simply computing a determinant. A non-exhaustive list is provided below.
\begin{itemize}
	\item Spanning trees in a graph can be counted by computing a determinant related to the Laplacian of the graph, a result known as the matrix-tree theorem \cite{Kir47}.
	\item Arborescences in a directed graph can be counted by computing a determinant related to the directed Laplacian \cite{Tutte48}.
	\item The number of perfect matchings in a planar graph can be computed as the Pfaffian (square root of the determinant) of an appropriately signed version of the adjacency matrix, a.k.a.\ the Tutte matrix \cite{Kas63}.
	\item The number of Eulerian tours in an Eulerian digraph is directly connected to the number of arborescences, and consequently the determinant related to the directed Laplacian \cite{BE87, ST41}.
	\item Given vectors $v_1,\dots,v_n\in \R^d$, the volume sampling distribution on subsets $S\in \binom{[n]}{d}$ can be defined as follows:
	\[ \P{S}\propto \det\parens*{\bracks{v_i}_{i\in S}}^2. \]
	The partition function of this distribution is simply $\det(\sum_{i} v_iv_i^\intercal)$ \cite[see, e.g.,][]{DR10}.
	\item The number of non-intersecting paths between specified terminals in a lattice, and more generally applications of the Lindström-Gessel-Viennot lemma \cite{Lin73,GV89}.
\end{itemize}

Efficient counting for these problems follows the polynomial-time computability of the associated determinants. In turn, one obtains efficient sampling algorithms for all of these problems; we remark that ot all of these problems are known to be self-reducible, but nevertheless ``easy'' slightly varied sampling to counting reductions exist for all of them.

While polynomial-time sampling for all of these problems has long been settled, we reopen the investigation of these problems by considering \emph{efficient parallel sampling} algorithms. We focus on the computational model of \PRAM, and specifically on the complexity classes \NC{} and \RNC{}. Here, a polynomially bounded number of processors are allowed access to a shared memory, and the goal is for the running time to be polylogarithmically bounded; the class \RNC{} has additionally access to random bits. Determinants can be computed efficiently in parallel, in the class \NC{} \cite{Csa75}, and as a result there are \NC{} counting algorithms for all of aforementioned problems. However, the sampling to counting reductions completely break down for parallel algorithms, as there seems to be an inherent sequentiality in these reductions.

Take spanning trees in a graph as an example. The classic reduction from sampling to counting proceeds as follows:
\begin{Algorithm*}
	\For{each edge $e$}{
		$A\leftarrow$ number of spanning trees containing $e$\;
		$B\leftarrow$ total number spanning trees\;
		Flipping a coin with bias $A/B$, decide whether $e$ should be part of the tree.\;
		Either contract or delete the edge $e$ based on this decision.\;
	}	
\end{Algorithm*}
Each iteration of this loop uses a counting oracle to compute $A, B$. However the decision of whether to include an edge $e$ as part of the tree affects future values of $A, B$ for other edges, and this seems to be the inherent sequentiality in this algorithm. The sampling-to-counting reduction for all other listed problems encounters the same sequentiality obstacle.

In this paper, we take a step towards resolving the mysterious disparity between counting and sampling in the parallel algorithms world. We resolve the question of sampling arborescences in weighted directed graphs, and as a special case spanning trees in weighted undirected graphs, using efficient parallel algorithms with access to randomness, a.k.a.\ the class \RNC{}.

We remark that the special case of sampling spanning trees in unweighted undirected graphs was implicitly solved by the work of \textcite{T95}, who showed how to simulate random walks in \RNC{}. When combined with earlier work of \textcite{Ald90,Bro89}, this algorithm would simulate a random walk on the graph, and from its transcript extract a random spanning tree. However, adding either weights or directions to the graph results in the need for potentially exponentially large random walks, which cannot be done in \RNC{}. Our work removes this obstacle.

\begin{theorem}\label{thm:main}
	There is an \RNC{} algorithm which takes a directed graph $G=(V, E)$ together with edge weights $w:E\to \R_{\geq 0}$ as input and outputs a random directed rooted tree $T$, a.k.a.\ an arborescence. The output $T$ follows the distribution
	\[ \P{T}\propto \prod_{e\in T} w(e). \]
\end{theorem}

\subsection{Related Work and Techniques}

There is a long line of research on algorithms for sampling and counting spanning trees and more generally arborescences. The matrix-tree theorem of \Textcite{Kir47} showed how to count spanning trees in undirected graphs, and later \textcite{Tutte48} generalized this to arborescences in digraphs. Somewhat surprisingly \textcite{Ald90} and \textcite{Bro89} showed that random spanning trees and more generally random arborescences of a graph can be extracted from the transcript of a random walk on the graph itself. The main focus of subsequent work on this problem has been on improving the total running time of sequential algorithms for sampling. After a long line of work \cite{Wil96, CMN96, KM09, MST15, DPPR17, DKPRS17}, \textcite{Sch18} obtained the first almost-linear time algorithm for sampling spanning trees. More recently \textcite{ALOV20} improved this to nearly-linear time. Many of these works are based on speeding up the Aldous-Broder algorithm. Our main result, \cref{thm:main}, is also built on the Aldous-Broder algorithm, but we focus on parallelizing it instead of optimizing the total running time.

No almost-linear time algorithm is yet known for sampling arborescences in digraphs, as opposed to spanning trees in undirected graphs. While sampling spanning trees has a multitude of application \cite[see][]{Sch18}, there are a number of applications for the directed graph generalization. Most notably, there is a many-to-one direct correspondence between Eulerian tours in an Eulerian digraph and arborescences of the graph. This correspondence, known as the BEST theorem \cite{BE87,ST41} allows one to generate random Eulerian tours by generating random arborescences \cite[see][]{Cre10}. We leave the question of whether the correspondence in the BEST theorem is implementable in \NC{} to future work, but note that sampling Eulerian tours has interesting applications in biology and sequence processing \cite{JAGM08, RBCD08}. We remark that, unlike directed graphs, generating random Eulerian tours of \emph{undirected} Eulerian graphs in polynomial time is a major open problem \cite{TV01}.

In a slightly different direction related to this work, \textcite{BD13} considered the space complexity of counting arborescences. They showed this problem is in $\mathsf{L}$ for graphs of bounded tree-width, obtaining algorithms for counting Eulerian tours in these digraphs as well.

\begin{figure}
	\begin{Columns}[Top]<31>
		\Column<10>
			\Tikz*[scale=0.8]{
				\begin{scope}[every node/.style={line width=1, draw=Black, fill=LightGray, inner sep=5, circle}]
					\node[draw=Orange, fill=LightOrange] (a) at (0, 0) {};
					\node (b) at (2, 0) {};
					\node (c) at (4, 0) {};
				\end{scope}
				\path[-stealth, line width=1, Black]
				(a) edge [bend left=20] (b)
				(b) edge [bend left=20] (a)
				(b) edge [bend left=20] (c)
				(c) edge [bend left=20] (b);
				\draw (1,.6) -- (1,.6) node [midway]  {W};
				\draw (1,-.6) -- (1,-.6) node [midway]  {W};
				\draw (3,.6) -- (3,.6) node [midway]  {1};
				\draw (3,-.6) -- (3,-.6) node [midway]  {1};
			}
		\Column<1>
		\Column<20>
			\Tikz*[scale=0.8]{
				\node at (10, 0.7) {};
				\begin{scope}[every node/.style={line width=1, draw=Black, fill=LightGray, inner sep=5, circle}]
					\node[draw=Orange, fill=LightOrange] (0) at (0, 0) {};
					\node (1) at (2, 0) {};
					\node (2) at (4, 0) {};
					\node (3) at (6, 0) {};
					\node (n) at (10, 0) {};
				\end{scope}
				\node (dots) at (8, 0) {\Large $\dots$};
				\path[-stealth, line width=1, Black]
				(0) edge [bend left=20] (1)
				(1) edge [bend left=20] (2)
				(2) edge [bend left=20] (3)
				(3) edge [bend left=20] (dots)
				(1) edge [bend left=20] (0)
				(2) edge [bend left=20] (0)
				(3) edge [bend left=20] (0)
				(n) edge [bend left=20] (0)
				(dots) edge [bend left=20] (n);
			}
	\end{Columns}
	\begin{Columns}[Top]<31>
		\Column<10>
			\caption{Starting from left it takes $\Theta(W)$ steps to cover.}\label{fig:weighted-undirected}
		\Column<1>
		\Column<20>
			\caption{A random walk started from the left node covers all $n$ nodes in time $\Theta(2^n)$.}\label{fig:unweighted-directed}
	\end{Columns}
\end{figure}

Perhaps in the first major result of its kind in the search for \RNC{} sampling algorithms, \textcite{Ten95} showed implicitly how to sample spanning trees in undirected graphs in \RNC{}. \textcite{Ten95} showed how to parallelize the simulation of a random walk on a graph; more precisely, he showed how to output a length $L$ trace of a walk on size $n$ Markov chains in parallel running time $\polylog(L, n)$ using only $\poly(L, n)$ many processors. The Aldous-Broder algorithm extracts an arborescence from the trace of a random walk by extracting the so-called first-visit edge to each vertex. If a random walk is simulated until all vertices are visited at least once, the trace of the random walk has enough information to extract all such first-visit edges. This allows \RNC{} sampling of arborescences in graphs where the number of steps needed to visit all vertices, known as the cover time, is polynomially bounded. While the cover time of a random walk on an undirected unweighted graph is polynomially bounded in the size of the graph \cite[see, e.g.,][]{LP17}, adding either weights or directions can make the cover time exponentially large. See \cref{fig:unweighted-directed,fig:weighted-undirected}.

We overcome the obstacle of exponentially large cover times, by taking a page from some of the recent advances on the sequential sampling algorithms for spanning trees. Instead of simulating the entirety of the random walk until cover time, we extract only the first-visit edges to each vertex. We use the same insight used in several prior works that once a region of the graph has been visited, subsequent visits of the random walk can be \emph{shortcut} to the first edge that exits this region. However, this involves a careful construction of a hierarchy of ``regions'', and bounding the number of steps needed inside each region to fully cover it. Unlike undirected graphs, in directed graphs arguing about the covering time of a region becomes complicated. We build on some of the techniques developed by \textcite{BPS18} to bound these covering times in the case of Eulerian digraphs. We then reduce the arborescence sampling problem for arbitrary digraphs to that of Eulerian digraphs.

\subsection{Overview of the Algorithm}

At a high-level, our algorithm first proceeds by reducing the problem to sampling an arborescence from an Eulerian graph. We then construct a ``loose decomposition'' of the Eulerian graph into weighted cycles. We then construct a hierarchy of vertex sets by starting with the empty graph and adding the weighted cycles one by one, from the lowest weight to the highest, adding the connected components in each iteration to the hierarchy. This results in a hierarchical clustering of vertices, with the intuitive property that once we enter a cluster, we spend a lot of time exploring inside before exiting. Our algorithm proceeds by ``simulating'' polynomially many jumps between the children of each cluster, before exiting that cluster. The resulting jumps are then stitched together at all levels of the hierarchy to form a partial subset of the transcript of a random walk.

 The shallow ``simulation'' of edges jumping across children of a cluster in the hierarchy can be done in \RNC{} by using a doubling trick as in the work of \textcite{Ten95}; in order to ``simulate'' $L$ jumping edges, we combine the first $L/2$ with the last $L/2$ in a recursive fashion. A naive implementation of this is not parallelizable however, as one needs to know \emph{where} the first $L/2$ jumps eventually land before ``simulating'' the rest. However since the number of locations is polynomially bounded, one can precompute an answer for each possible landing location in parallel to simulating the first $L/2$ edges. Once shallow jumps are simulated, our algorithm then proceeds by stitching these jumping edges with jumping edges of one level lower (i.e., inside of the children clusters), and so on. Again, a naive implementation of this stitching is not in \RNC{}, but we employ a similar doubling trick and an additional caching trick to parallelize everything. We then argue that with high probability all of the edges extracted by this algorithm contain the first-visit edges to every vertex.

\section{Preliminaries}
\label{sec:prelim}


We use the notation $[n]$ to denote the set $\set{1,. . .,n}$.  We use the notation $\Tilde{O}(.)$ to hide polylogarithmic factors. When we use the term high probability, we mean with probability $1-\poly(\frac{1}{n})$, where $n$ is the size of the input, and the polynomial can be taken arbitrarily large by appropriately setting parameters.

\subsection{Graph Theory Notations}

When a graph $G= (V,E)$ is clear from the context, we use $n$ to refer to $\card{V}$ and $m$ to refer to $\card{E}$.

For a subset of vertices $S \subset V$, $\delta^+(S)$ denotes the set all incoming edges $\set{e = (u,v)\given v \in S, u \notin S}$. Similarly, $\delta^-(S)$ will denote the set of all outgoing edges $\set{e = (u,v)\given u \in S, v \notin S}$. Lastly, let $G(S)$ denote the subgraph induced by $S$.

\begin{definition}
Given a weighed graph $G=(V, E)$ with weights $w:E\to \R_{\geq 0}$, a subset of the vertices $S \subseteq V$ is said to be \textit{strongly connected} by edges of weight $w_c$ if for every $s, t, \in S$, there exists a path of edges inside $S$ connecting $s$ to $t$ such that every edge $e$ of the path has weight $w(e)\geq w_c$.
\end{definition}

\begin{definition}
An \emph{arborescence} is a directed graph rooted at vertex $r$ such that for any other vertex $v$, there is exactly one path from $v$ to $r$.
\end{definition}
Arborescences are also known as directed rooted trees and are a natural analog of spanning trees in directed graphs. When a background graph $G=(V, E)$ is clear from context, by an arborescence we mean a subgraph of $G$ that is an aborescence.

Given digraph $G=(V, E)$ together with a weight function $w:E\to \R_{> 0}$, the random arborescence distribution is the distribution that assigns to each arborescence $T$ probability
\[ \P{T}\propto \prod_{e\in T} w(e). \]
When $w$ is not given, it is assumed to be the constant $1$ function, and the random arborescence distribution becomes uniformly distributed over all arborescences. We view an undirected graph $G=(V, E)$ as a directed graph with double the number of edges, with each undirected edge producing two directed copies in the two possible directions. It is easy to see that the weighted/uniform random arborescence distribution on the directed version of an undirected graph is the same as the weighted/uniform spanning tree distribution on the original graph; viewing each arborescence without directions yields the corresponding spanning tree.

We call a (possibly weighted) digraph $G=(V, E)$ Eulerian if \[\sum_{e\in \delta^+(v)}w(e)=\sum_{e\in \delta^{-}(v)}w(e)\] for all vertices $v$. Note that the directed version of an undirected graph is automatically Eulerian.

\subsection{Random Walks and Arborescences}
Our approach centers around the Aldous-Broder algorithm on directed graphs \cite{Ald90,Bro89}. Their work reduces the task of randomly generating a spanning tree or arborescence to simulating a random walk on the graph until all vertices have been visited. While most famously known for generating spanning trees, this algorithm can be used to generate arborescences as well.

A weighted digraph defines a natural random walk. This is a Markov process $X_0, X_1, \dots$, where each $X_i$ is obtained as a random neighbor of $X_{i-1}$ by choosing an outgoing edge with probability proportional to its weight, and transitioning to its endpoint.
\[ \P{X_i\given X_{i-1}}\propto w(X_{i-1}, X_i). \]
A stationary distribution $\pi$ is a distribution on the vertices such that if $X_0$ is chosen according to $\pi$, then all $X_i$ are distributed as $\pi$. Under mild conditions, namely strong connectivity and aperiodicity, the stationary distribution is unique.

A random walk on an undirected graph is known to be time-reversible. That is if $X_0$ is started from the stationary distribution, then $(X_i, X_{i+1})$ is identical in distribution to $(X_{i+1}, X_i)$. This does not hold for directed graphs. However the time reversal of the process $\dots, X_i, X_{i+1}, \dots$ corresponds to a random walk on a different digraph.
\begin{definition}
	For a weighted digraph $G=(V, E)$ with weights $w:E\to\R_{\geq 0}$ and stationary distribution $\pi$, define the time-reversal to be the graph $G'=(V, E')$ on the same set of vertices, with edges reversed in direction, and weights given by
	\[ w'(v, u)=\pi(u)w(u, v)/\pi(v). \]
\end{definition}
The random walk on $G'$ shares the same stationary distribution $\pi$ as $G$. The random walk on $G'$ (started from $\pi$) is identical in distribution to the time-reversed random walk on $G$ (started from $\pi$); see \cite{LP17}.

\begin{theorem}[{\cite{Ald90, Bro89}}]\label{theo:aldousBroder}
Suppose that $G= (V,E)$ is a strongly connected weighted graph, with weights given by $w:E \to \R_{>0}$. Perform a time-reversed random walk, starting from a vertex $r$, until all vertices are visited at least once. For each vertex $v \in V \setminus \{r\}$, record the edge used in the random walk to reach $v$ for the first time. Let $T$ be the collection of all these first-visit edges (with directions reversed). Then $T$ is an arboresence of $G$ rooted at $r$, and
\[ \P{T} \propto \prod_{e \in T}w(e) \]
\end{theorem}

This allows us to sample $r$-rooted arborescence from a digraph by performing a random walk on the time-reversal. It is also known that among all arborescences, the total weight of those rooted at $r$ is proportional to $\pi(r)$, where $\pi$ is the stationary distribution \cite{Ald90, Bro89}. Therefore to resolve \cref{thm:main}, it is enough to show how to sample $r$-rooted arborescences.
\begin{theorem}
	There is an \RNC{} algorithm which takes a directed graph $G=(V, E)$ together with edge weights $w:E\to\R_{\geq 0}$ and a root vertex $r\in V$ as input and outputs a random $r$-rooted arborescence, where
	\[\P{T}\propto \prod_{e\in T} w(e). \]
\end{theorem}
We can fix the root because the stationary distribution $\pi$ can be computed in \NC{}, by solving a system of linear equations.

\subsection{Parallel Algorithms}

In this paper we consider parallel algorithms run on the \PRAM{} model, and we will construct algorithms to show that sampling problems are in the class \RNC{}, the randomized version of \NC{}.  In this class, we are  allowed  to  use  polynomially  (in  the  input  size)  many  processors  who  share  access  to  a  common random access memory, and also have access to random bits. The running time of our algorithms must be polylogarithmic in the input size.

While our work hinges upon the simulation of a random walk, a process which is inherently sequential, for polynomially many steps, such a task is known to be in \RNC{} through the use of a ``doubling trick.''
\begin{theorem}[\cite{Ten95}] \label{theo:parallelWalks}
Suppose that $G = (V,E)$ is a directed graph, with weights given by $w:E \to \R_{\geq 0}$.  The transcript of a random walk starting from a given vertex $v_0 \in V$ and running for $L$ time steps can be produced in the \PRAM{} model using $\poly(\card{E},L)$ processors in $\polylog(\card{E},L)$ time.
\end{theorem}

\Cref{theo:aldousBroder} and \cref{theo:parallelWalks} naturally lead to a parallel algorithm that randomly samples arborescences using $\poly(\card{E},L)$ processors and $\polylog(\card{E},L)$ time where $L$ is the cover time of the digraph. However, $L$ can be exponential in the size of the graph as the cover time depends closely on the edge weights of a graph. In directed graphs, the cover time can be exponential even if we do not allow weights. See \cref{fig:unweighted-directed,fig:weighted-undirected}.

Luckily, \cref{theo:aldousBroder} only needs the first-visit edges to produce the random arborescence. This insight has been heavily used to improve the running time of sequential algorithms for sampling spanning trees \cite{KM09, MST15, Sch18}, by \textit{shortcutting} the random walk. Here we use the same insight to design an \RNC{} algorithm. Our algorithm identifies a hierarchy of clusters of vertices $S \subseteq V$ and will only simulate incoming and outgoing edges of each cluster as opposed to the entire walk. To do this, we will use a well-known primitive that the probability of entering or exiting a cluster through any edge can be computed using a system of linear equations.
\begin{lemma}[{\cite[see, e.g.,][]{Sch18}}] \label{lem:shortcuttingPrimitive}
Given a set $S\subset V$ and vertex $v\in S$, the probability of a random walk started at $v$ exiting $S$ through any particular edge $e\in \delta^-(S)$ can be computed by solving a system of linear equations involving the Laplacian.
\end{lemma}

Some of the most powerful primitives for \NC{} algorithms come from linear algebra. In particular, multiplying matrices, computing determinants, and inverting matrices all have \NC{} algorithms \cite{Csa75}. Combining this with \cref{lem:shortcuttingPrimitive} we obtain the following folklore result.

\begin{lemma}
\label{lem:NCshortcutting}
Given a set $S \in V$ and vertex $v \in S$, the probability of a random walk started at $v$ exiting $S$ through any particular edge $e \in \delta^-(S)$ can be computed in \NC{}.
\end{lemma}

\subsection{Schur Complements}
A key tool which we shall use in our analysis of random walks is Schur complements. 
\begin{definition}
For any weighted digraph $G = (V, E)$, for any subset of the vertices $S \subseteq V$, we define the graph
$G_{S}$  to be the \textit{Schur complement} of $\bar{S}$. $G_{S}$ is the directed graph formed by starting with the induced sub-graph of $S$. Then, for every pair of vertices ($u,v$) in $S$, add an edge from $u$ to $v$ with weight $\deg(u)\P{u \xrightarrow{} v}$, where $\P{u \xrightarrow{} v}$ is the probability that a random walk on $G$ currently at $u$ exits $S$ with the next move and its first return to $S$ is at the vertex $v$. 
\end{definition}

Then, it is easy to see that for any starting vertex $v_0 \in S$, the distribution over random walk transcripts in $G_S$ starting at $v_0$ is the same as the distribution over random walk transcripts in $G$ starting at $v_0$
with all vertices outside of $S$ removed. We also observe that:

\begin{proposition} \label{prop:eulerianSchure}
For any Eulerian weighted digraph $G = (V, E)$, for any subset of the vertices $S \subset V$, the Schur complement $G_S$ is also an Eulerian weighted digraph and the degree of any vertex in $S$ is the same in $G_S$ as in $G$.
\end{proposition}

	\section{Reduction to Eulerian Graphs}
\label{sec:reduction}

A key part of our algorithm is the analysis of the time it takes for random walks to cover regions of a digraph. Cover times are in general difficult to analyze on directed graphs, but the situation becomes much easier on Eulerian graphs. For more precise statements, see \cref{sec:exploration}. As a first step in our algorithm, we show how to reduce the design of sampling algorithms on arbitrary digraphs to Eulerian ones.

We can reduce the problem of sampling an arborescence of digraph $G$ to sampling a random arborescence rooted at vertex $r$ on strongly connected Eulerian digraph $G''$ as follows in \cref{alg:reduction}. We begin by selecting the vertex at which the arborescence is rooted, using the Markov chain tree theorem \cite{Ald90, Bro89}.
\begin{theorem}[The Markov chain tree theorem] Let $G=(V,E)$ be a weighted directed graph. Assume that the natural Markov chain associated with $G$ has stationary distribution $\pi=(\pi_1, \dots , \pi_n)$. Then we have the following relationship between $\pi$ and the arborescences of $G$: 
\[\pi_i = \frac{\sum\set{w(T)\given T \text{ arborescence rooted at $i$}}}{\sum\set{w(T)\given T \text{ arborescence}}},\]
where $w(T)$ is product of edge weights in $T$.
\end{theorem}
 Finding the stationary distribution of a random walk on some graph $G$ reduces to solving a system of $n$ linear equations and can be done in \NC{}. 

After the root $r \in V$ of an arborescence is selected, we can add outgoing edges from $r$ to graph $G$ to guarantee that $G$ is strongly connected. As no arborescence rooted at $r$ will contain these new edges, the addition of such edges does not affect the distribution of arborescences rooted at $r$. Then, we note that for any vertex $v$, if we multiply the weights of all edges in $\delta^-(v)$ by some constant, the distribution of arborescences rooted at any vertex does not change. We can therefore rescale edges to ensure that the resulting graph is Eulerian.

\begin{Algorithm}

\caption{Reduction to strongly-connected Eulerian graphs \label{alg:reduction}}

Compute the stationary distribution $\pi(v)$ of a random walk on $G$ \;
Choose vertex $r$ to be the root of the arborescence with probability $\pi(r)$\;
$G' = (V, E') \xleftarrow{} G$\;
\For{$v \in V \setminus \{r\}$ \textbf{\textup{in parallel}}}{
    \If{$e' = (r, v) \notin E$}{
        Add  $e'$ to $E'$ with weight $w'(e) = 1$\;
    }
}
Compute the stationary distribution $\pi'(v)$ of a random walk on $G'$\; 
$G'' = (V, E'') \xleftarrow{} G'$ \;
\For{$v \in V$}{
    \For{$e = (v, u) \in E''$}{
        $w''(e) \xleftarrow{} \frac{w'(e)}{\deg_{out}'(v)}\pi'(v)$ \;
    }
}
Randomly sample an arborescence rooted at $r$ from $G''$\;

\end{Algorithm}

\begin{proposition}
Graph $G''$ is Eulerian.
\end{proposition}
\begin{proof}
We have that, for all $v$ in $G'$
\[\sum_{(u, v) \in \delta^+(v)} \pi'(u)\frac{w'(u, v)}{\deg_{\text{out}}'(u)} = \pi'(v)\]
Then, considering $G'':$
\begin{align*}
\deg_{\text{in}}''(v) &= \sum_{(u, v) \in \delta^+(v)} w''(u, v) = \sum_{(u, v) \in \delta^+(v)} \frac{w'(u, v)\pi'(u)}{\deg_{\text{out}}'(u)} = \pi'(v)\\
&= \pi'(v) \sum_{(v, x) \in \delta^- (v)} \frac{w'(v, x)}{\deg_{\text{out}}'(v)} = \sum_{(v, x) \in \delta^- (v)} w''(v, x) = \deg_{\text{out}}''(v) \qedhere
\end{align*}
\end{proof}

Thus, in \RNC{}, the task of randomly sampling an arborescence in some digraph can be reduced to the task of randomly sampling an arborescence rooted at some vertex $r$ from a strongly connected Eulerian digraph.
	\section{Cycle Decomposition}
\label{sec:decomposition}
\newcommand{\T}{\mathcal{T}}
\newcommand{\sS}{\mathcal{S}}

Because \Cref{theo:aldousBroder} only needs the first-visit edges to produce a random arborescence, our main idea is to create a hierarchical clustering of the graph that we call $\T$. Every element $\sS \in$ $\T$ in the clustering is a subset of $V$ which is strongly connected by edges of relatively high weight compared to the edges in $\delta^-(\sS)$. Intuitively, a random walk will spend much time traversing edges inside a cluster, before venturing out. Aside from first visit edges, any other traversals do not need to be simulated, thus much of the work can be avoided by only simulating the first few edges which enter or exit a cluster. Consider the pseudocode in \cref{alg:decomposition} for generating one such decomposition:

\begin{Algorithm}
\caption{Generating a hierarchical decomposition \label{alg:decomposition}}

\For{edge $e \in G$ \textbf{\textup{in parallel}}}{
    Find a cycle $C_e$ of edges such that $e \in C_e$ and for every $e' \in C, w_{e'} \geq \frac{w_e}{m}$ using \cref{alg:find-cycle} \;
}

Sort the cycles found by minimum edge weight in decreasing order $C_1, C_2, \dots, C_m$ \;
$\mathcal{T} \xleftarrow{}  \{\}$\;
\For{i = 1, 2, \dots m \textbf{\textup{in parallel}}}{
    Find $\mathcal{S}_{i,1}, \mathcal{S}_{i,2} \dots \mathcal{S}_{i,j}$, the vertex sets of the connected components in $(V, C_1 \cup C_2 \cup \dots \cup C_i)$ \;
    Add all elements found to $\mathcal{T}$ \;
}

\For{$\mathcal{S}_{i,j} \in \mathcal{T}\setminus S_{m, 1}$ \textbf{\textup{in parallel}}}{
    Find the node of the form $\sS_{i+1, k}$ such that $\sS_{i, j} \subset \sS_{i+1, k}$ and assign $\sS_{i+1, k}$ as the parent of $\sS_{i, j}$
}
Contract duplicate nodes in $\mathcal{T}$\;

\end{Algorithm}

\begin{Algorithm}
\caption{\texttt{FindCycle}$(e = (e_0, e_1), G = (E, V))$}\label{alg:find-cycle}
Construct $G' = (V, E')$, where $E' = \{e' \in E: w_{e'} \geq \frac{w_e}{m}\}$\;

\For{$v_1, v_2 \in V$ \textbf{\textup{ in parallel}}}{
    \If{$(v_1, v_2) \in E'$}{$R(v_1, v_2, 1) \xleftarrow{} $True \;
    }
}

\For{$v \in V$ \textbf{\textup{in parallel}}}{
    $R(v, v, 0) \xleftarrow{}  $True\;
    $R(v, v, 1) \xleftarrow{}  $True\;
}

\For{$l \in \{2,4,  \dots 2^{\ln{n}}\}$}{
    \For{$v_1, v_2, v_3 \in V$ \textbf{\textup{ in parallel}}}{
        \If{$R(v_1, v_2, l/2)$ \textup{and} $R(v_2, v_3, l/2)$}{
            $R(v_1, v_3, l) \xleftarrow{}  $True \;
        }
    }
}
$c_0 \xleftarrow{}  e_1, c_{2^{\lceil \ln{n} \rceil}} \xleftarrow{}  e_0 $\;
\For{$l \in \{2^{\lceil \ln{n} \rceil - 1}, 2^{\lceil \ln{n} \rceil - 2}, \dots 1 \}$}{
    \For{$i \in [2^{\lceil \ln{n} \rceil}]$ \textup{such that $l$ is the largest power of $2$ which divides $i$} \textbf{\textup{ in parallel}}}{
        Set $c_i$ to be some vertex  such that $R(c_{ i - l}, c_i, l)$ and $R(c_{ i}, c_{i+l}, l)$ \;
    }
}
Prune the transcript of vertices $\{c_0, c_1, \dots, c_{2^{\lceil \ln{n} \rceil}}\}$ to remove any duplicates and return the resulting cycle. \;

\end{Algorithm}

We use \cref{alg:find-cycle} to construct a cycle of comparably high-weight edges containing a given edge $e$. The collection of these cycles is a ``loose decomposition'' of the Eulerian graph into cycles. While an Eulerian graph can be decomposed into cycles in general, we do not know if this can be done in \NC{}. Instead we find a collection of cycles whose sum and $\poly(m)$ times it sandwich the Eulerian graph.

For each edge $e$, \cref{alg:find-cycle} successfully returns a cycle since,
\begin{proposition}
In directed weighted Eulerian graph $G = (V, E)$ for any edge $e \in E$, there exists a cycle of edges $C_e$ such that $e \in C_e$ and for all $e' \in C_e, w_e' \geq \frac{w_e}{m}$
\end{proposition}
\begin{proof}
Note that any Eulerian digraph can be decomposed into at most $m$ cycles where each cycle contains edges of uniform weight. Then, for any edge $e$, one of these cycles must have weight at least $\frac{w_e}{m}$. 
\end{proof}
Given that a cycle is found containing every edge in $G$, it follows that \Cref{alg:decomposition} successfully generates a hierarchical decomposition which satisfies: 
\begin{proposition} $\mathcal{T}$ is a laminar family, i.e. has the following properties:
\begin{enumerate}
    \item Every node $\mathcal{S}$ of $\mathcal{T}$ is a subset of $V$ and is the vertex set of a connected component of $G$.
    \item Children of any node in $\mathcal{T}$ are proper subsets of it.
    \item $\card{\mathcal{T}} \leq 2n.$

\end{enumerate}
\end{proposition}
Given this clustering $\T$, for any $\sS \in \T$, we say an edge $e \in E$ jumps between children of $\sS$ if $\sS$ is the lowest node in $\T$ which contains both endpoints of $e$. We call $e$ a \textit{jumping edge} of $\sS$.

In addition, we note that 

\begin{lemma}\label{lem:connectedness}
For any $\sS \in $ $\T$, let $w_{\text{max}}(\sS)$ be the maximum weight among the jumping edges of $\sS$. $G(\sS)$ is strongly connected by edges of weight $\frac{w_{\text{max}}(\sS)}{m}$. 
\end{lemma}

\begin{proof}
We prove by contradiction. Assume that there is some edge $e$ jumping between children of $\sS$ and two vertices $u, v \in \sS$ such that every path between $u, v$ contains an edge of weight less than $\frac{w_e}{m}$. Let the cycles which connect $\sS$ in \cref{alg:decomposition} be $C_1, \dots C_k$, with $C_k$ being the cycle with an edge of minimum weight. Then, consider the path of edges connecting $u$ to $v$ only containing edges in these cycles. The weight of every edge in the path must be greater than $w(C_k)$, the weight of the minimum weight edge in $C_k$. Thus by our assumption we must have that $\frac{w_e}{m} > w(C_k)$.

But now consider $C_e$, the cycle found containing edge $e$ with minimum edge weight $w(C_e) \geq \frac{w_e}{m}$. As $w(C_e) > w(C_k)$, this cycle must have been used to construct a child of $\sS$ in $\T$. This implies the $e$ is contained in a child of $\sS$ and is therefore not a jumping edge of $\sS$, a contradiction.
\end{proof}

\Cref{lem:connectedness} intuitively states that the nodes in $\T$ are tightly connected by edges of relatively high weight compared to incoming or outgoing edges. This will be critical in bounding the number of edges which need to be simulated.
	\section{Parallel Sampling}
\label{sec:sampling}

For clarity, a partially sampled random walk shall be represented as an alternating series of clusters and edges: $\sS_1, e_1, \sS_2, e_2, \dots$, where each $\sS_i \in \T$ and $e_i \in \delta^-(\sS_i) \cap \delta^+(\sS_{i+1})$. The sequence has a natural interpretation of a random walk in $\sS_i$ which then exits $\sS_i$ through $e_i$ to continue in $\sS_{i+1}$ before traversing $e_{i+1}$ and so on.

Once we have a hierarchical decomposition $\T$ found in \cref{sec:decomposition}, a natural way to use such a decomposition is to simulate a random walk as follows: 
\begin{enumerate}
    \item Simulate just the edges jumping between high-level clusters of a random walk, producing a transcript of the form $\sS_1, e_1, \sS_2, e_2, \dots$.
    \item For each $i$, recursively simulate a random walk on $\sS_i$ conditioned on exiting $\sS_i$ though edge $e_i$.
\end{enumerate}
While this recursive procedure as described is not in \RNC{}, using doubling tricks and memoization, we show how to perform this in \RNC{}. One can view this as a generalization of the doubling trick used by \textcite{Ten95}.

\subsection{Simulating Jumping Edges}
To simulate edges jumping across children of a node, we generalize the techniques of \cite{Ten95} to the following: 

\begin{theorem}
Suppose that $G = (V,E)$ is a directed graph, with weights given by $w:E \to \R_{\geq 0}$. Let $\sS \subseteq V$ be the disjoint union of $\sS_1, \sS_2, \dots \sS_k$, where $\sS$ and each $\sS_i$ is strongly connected.  The transcript of first $L$ edges jumping across different $\sS_i$ in a random walk starting from a given vertex $v_0 \in \sS$ and conditioned on exiting $\sS$ though some edge $e \in \delta^-(\sS)$ can produced in the \PRAM{} model using $\poly(\card{E},L)$ processors and $\polylog(\card{E},L)$ time.
\end{theorem}

\begin{proof}
We denote this function $\texttt{JumpingEdges}(\sS, v_0, e, L)$. See \cref{alg:jumpingEdges} for pseudocode. Note that we can condition a walk on $\sS$ to exit via some edge $e \in \delta^-(\sS)$ by considering a random walk on the graph induced by $\sS$ with an additional edge $e$ and an absorbing state at the endpoint of $e$. For any vertex $v \in \sS$, let $\sS_v$ denote the unique $\sS_i$ which contains $v$.

From \cref{lem:NCshortcutting}, for any $v \in \sS$, the probability that a random walk currently at $v$ exits $\sS_v$ though each edge in $\delta^-(\sS_v)$ can be computed in \NC{}. Thus, after using $\poly(\card{E},L)$ processors and $\polylog(\card{E},L)$ time to compute all such probabilities for every vertex $v \in \sS$, we can assume access to a deterministic function \texttt{Next}$(v, X)$ which takes a vertex $v \in \sS$ and number $X \in [0,1]$ and outputs and edge $e \in \delta^-(\sS_v)$ such that if $X$ is chosen uniformly from $[0,1]$, then the probability that \texttt{Next}$(v, X)$ returns some edge $e \in \delta^-(\sS_v)$ is the probability a random walk currently at $v$ exits $\sS_v$ through edge $e$.

Given this function, if we generate independent uniformly random numbers $X_1,. . .,X_L \in [0, 1]$, the transcript of the jumping edges in random walk can be extracted by starting at $v=v_0$ and repeatedly applying \texttt{Next}$(v,X_i)$ to get the next vertex $v$, for $i \in [L]$. Of course, a naive implementation does not leverage parallelism, so we instead use a doubling trick.

We compute the values \texttt{End}$(v,t, l)$ which would be the final edge jumping between children of $\sS$ when we run a random walk starting at $v$, and applying \texttt{Next} with random inputs $X_t,X_{t+1},. . .,X_{t+l-1}$. For $l=1$, these are 1-step random walks and we use \texttt{Next} to compute them. For simplicity, we allow the first argument of $\texttt{End}$ to also be an edge, in which case the random walk starts at the ending vertex of that edge. Then we use the following identity, which allows us to compute \texttt{End} values for a particular $l$ from \texttt{End} values for $\frac{l}{2}$.
\[\texttt{End}(v, t, l) = \texttt{End}(\texttt{End}(v, t, l/2), t, l/2)\]

This is because we can break an $l$-step random walk into two $l/2$-step random walks. So we can compute all \texttt{End} values when $l$ ranges over powers of 2, in $\log(L$) steps.

Finally to compute the edge taken by the random walk at any time $l=1,. . .,L$, we simply write down $l$ as a sum of powers of 2, and repeatedly use the \texttt{End} function to compute the $l$-th edge of the random walk which jumps between children of $\sS$. 
\end{proof}

\begin{Algorithm}[h]

\caption{\texttt{JumpingEdges}$(\sS, v_0, e_{end}, L)$ \label{alg:jumpingEdges}}

\For{$t \in \{1, 2, \dots L\}$ \textbf{\textup{in parallel}}}{
    Let $X_t$ be a uniformly random number sampled from $[0,1]$ \;
    \For{$v \in S$ \textbf{\textup{in parallel}}}{
        $\texttt{End}(v, t, 1) \xleftarrow{}  \texttt{Next}(v, S)$\;
    }
}
\For{$l \in \{2, 4, 8, \dots, 2^{\lfloor \ln L \rfloor }\}$}{
    \For{$v \in S, t \in \{0, 1, \dots L - 1 - l\}$ \textbf{\textup{in parallel}}}{
        
        \texttt{End}$(v,t, l)$ $\xleftarrow{}$   \texttt{End}$(\texttt{End}(v,t, \frac{l}{2}),t, \frac{l}{2}) $ \;
        
    }
}
\For{$l \in {1, 2, \dots L}$ \textbf{\textup{in parallel}}}{
    Decompose $l$ as a sum of subset $S \subset \{1, 2, \dots 2^{\lfloor \ln L \rfloor - 1} \}$\;

    $e \xleftarrow{} (*, v_0)$\;
    $t \xleftarrow{} 0$\;
    
    \For {$s \in S$}{
        $e \xleftarrow{}  \texttt{End}(e, t, s)$ \;
        $t \xleftarrow[]{} t + s$\;
    }
    $e_l \xleftarrow{} e$\;
    $v_l \xleftarrow{}$ the ending vertex of $e_l$ 
}
\If{$e_k = e_{\text{end}}$ \textup{for some} $k$}{
    return $\sS_{v_0}, e_1, \sS_{v_1}, e_2, \dots \sS_{v_{k-1}}, e_k = e_{\text{end}}$\;
}
\Else{
    return $\sS_{v_0}, e_1, \sS_{v_1}, e_2, \dots, \sS_{L-1}, e_L, \sS, e_{\text{end}}$\;
}

\end{Algorithm}

\begin{remark}
For clarity of exposition, we allow sampling numbers from $[0, 1]$. However sampling a number with polynomially many bits is enough for our purposes. In our algorithms, we only need to compare our samples $X_t$ with deterministic numbers derived from the input. To this end, we can simply sample the first $N$ digits of the binary expansion of $X_t$ for some large $N$. The probability that the comparison of $X_t$ and a fixed number is not determined from the first $N$ digits is $2^{-N}$. By taking $N$ to be polynomially large, the probability of failure will be bounded by an exponentially small number. If one insists on avoiding even this small probability of failure, we can continue sampling digits of $X_t$ whenever we run into a situation where the first $N$ are not enough.  Since the probability of running into failure is very small, the overall expected running time still remains polylogarithmic.
\end{remark}

\subsection{Extracting First-Visit Edges}

For a cluster, once we know which edges jump across its children, we can recursively fill in the transcript of the walk. This is because we now have a similar subproblem for each child node, where we have a starting vertex, and a prespecified exit edge. However, a naive implementation of this is not in \RNC{}.

Instead our strategy is to again use a doubling trick.  Consider a node $\sS$ of $\T$ and a vertex $v \in \sS$.  Using the algorithm \texttt{JumpingEdges}, we can extract the edges that jump across children of $\sS$. We can then extend these to include edges that jump across children of children of $S$ by more applications of \texttt{JumpingEdges}. We will construct a function \texttt{AllEdges}$(S,v,e,L,l)$ that besides the arguments to \texttt{JumpingEdges} also takes an integer $l \geq 1$.  The goal of this function is to extract from the transcript of a random walk started at $v$ and conditioned on exiting $\sS$ through $e$, the first $L$ edges which jump between descendants of depth at most $l$ below $\sS$ in $\T$. Then, each intermediate node in the transcript is either an individual vertex, an $l$th descendant of $\sS$, or a cluster for with $L$ edges jumping between its children have already been found. We will show that for polynomially large $L$, this transcript will contain all first-visit edges with high probability. 

A natural approach for computing \texttt{AllEdges} is as follows. We already know how to compute \texttt{AllEdges} for $l=1$; that is just a call to \texttt{JumpingEdges}. Once the value of \texttt{AllEdges} has been computed for all settings of parameters for a particular $l$, we can compute it for $2l$ by the following procedure:

\begin{enumerate}
    \item  Let the transcript returned by \texttt{AllEdges}$(S,v,e,L,l)$ be $\sS_1, e_2,\sS_2, e_2, \dots , \sS_k,e_k$.
    \item For each intermediate node $\sS_i$ that is not a single vertex, if the transcript does not already contain $L$ edges jumping between children of $\sS_i$, let $u$ be the endpoint of $e_{i-1}$ and replace $\sS_i, e_i$ with \texttt{AllEdges}$(\sS_i,u,e,L,l)$
    \item Trim the transcript to only contain at most the first $L$ edges jumping between children of any node
\end{enumerate}

There is one key problem with this approach. We cannot precompute \texttt{AllEdges}$(\cdot,\cdot,\cdot,\cdot,l)$ and use the same precomputed answer to fill in the gaps for larger values of $l$. Each time we need to use \texttt{AllEdges} for a particular setting of its input parameters, we really need to use fresh randomness; otherwise the transcript we extract will not be from a true random walk.

We resolve this by using a caching trick. Instead of computing \texttt{AllEdges}$(S,v,e,L,l)$ for each setting of parameters once and reusing the same output for larger subproblems, we compute $M$ possible answers for a large enough $M$ and store all $M$ answers.  For larger subproblems, every time that we need to use the values of a smaller subproblem, we randomly pick one of the stored $M$ answers.  We will inevitably reuse some of the answers in this process; however when we restrict our attention to the unraveling of a particular answer for a subproblem there is a high probability of not having reused any answers. We will formalize this in \cref{lem:birthday}.

Pseudocode for computing \texttt{AllEdges} can be found in \cref{alg:allEdges}. In the end we use \\ \texttt{AllEdges}$(V,v, \emptyset, L,\card{V})$ to extract a list of edges, and with high probability all first-visit edges will be among this list

\begin{Algorithm}

\caption{Computing the random walk transcript extracts \label{alg:allEdges}}

\For{$\mathcal{S} \in \mathcal{T}, e \in \delta^-(\mathcal{S})\cup \{\emptyset\}, v \in \mathcal{S}, i \in [M]$ \textbf{\textup{in parallel}}}{
    
    \texttt{AllEdges}$(\mathcal{S}, v, e, L, 1)[i] \xleftarrow{}  \texttt{JumpingEdges}(S, v, e, L)$ \;

}
\For{$l \in \{1, 2, \dots 2^{\lceil \lg L \rceil-1}\}$}{
    \For{$\mathcal{S} \in \mathcal{T}, e \in \delta^-(\mathcal{S})\cup \{\emptyset\}, v \in \mathcal{S}, i \in [M]$ \textbf{\textup{in parallel}}}{
        \If{\normalfont{$l$ is greater than depth of the deepest child of $\mathcal{S}$}}{
            \texttt{AllEdges}$(\mathcal{S}, v, e, L, 2l)[i] \xleftarrow{}  \texttt{AllEdges}(\mathcal{S}, v, e, L, l)[i]$\;
        }
        \Else{
            Let $\sS_1, e_1, \sS_2, e_2, \dots \sS_k$ be the output of $\texttt{AllEdges}(\mathcal{S}, v, e, L, l)[i]$\;
            \For{$j \in [k]$ \textbf{\textup{in parallel}}}{
                \If{\normalfont{the transcript does not already contain $L$ edges jumping between children of $\sS_j$}}{
                    Replace $S_i, e_i$ with a randomly chosen instance of  \texttt{AllEdges}$(\mathcal{S}_j, v', e_j, L, l)$ where $v'$ is the endpoint of $e_{j-1}$ \;
                    
                }
            }
            Trim the resulting transcript to only contain at most the first $L$ edges jumping between children of any node. Save the resulting transcript as a solution to
            \texttt{AllEdges}$(\mathcal{S}, v, e, L, 2l)[i]$\;
        }
    }
}
\If{\normalfont{the final transcript of \texttt{AllEdges}$(V, v, \emptyset, L, \card{V})$} contain multiple subsequences which depend on the same call to $\texttt{JumpingEdges}$}{
    Replace all but the first subsequence with a freshly sampled transcript\;    
}

\end{Algorithm}

\begin{lemma} \label{lem:birthday}
For $M = \poly(L)$, when we unravel the recursion tree for the computation of the value \texttt{\textup{AllEdges}}$(V,v,\emptyset,L,\card{V})$, no stored answer of \texttt{\textup{AllEdges}} for any of the subproblems will be used more than once with high probability.
\end{lemma}
\begin{proof}
Note that for any any $\sS \in T, v \in \sS, e \in \delta^-(\sS)$, the recursion tree for an answer to \texttt{\textup{AllEdges}}$(V,v,\emptyset,L,\card{V})$ will rely on at most $L$ calls to $\texttt{JumpingEdges}(\sS, v, e, L)$ since each call to $\texttt{JumpingEdges}(\sS, v, e, L)$ corresponds to a subsequence of the final transcript which ends in $e$, and $e$ can only occur $L$ times in the transcript by the definition of $\texttt{AllEdges}$. As each sample of $\texttt{JumpingEdges}(\sS, v, e, L)$ is chosen uniformly at random from $M$ samples, for polynomially large $M$, the probability of the same sample being chosen twice can be made small and of the form $\frac{1}{\poly(n)}$. Lastly via a union bound over polynomially many possible combinations of arguments for \texttt{JumpingEdges}, we have that for polynomially large $M$, which high probability no stored answer of \texttt{JumpingEdges} or \texttt{AllEdges} will be used more than once. 
\end{proof}

Putting everything together, we have shown that:
\begin{theorem}\label{lem:polyAllEdges}
Suppose that $G = (V,E)$ is a directed graph, with weights given by $w:E \xrightarrow{} \R_{>0}$. Let $\T$ denote a hierarchical decomposition of $G$. The transcript of a random walk starting at $v_0 \in V$ which contains first $L$ edges jumping between children of $\sS$ for any $\sS \in \T$ can produced in the \PRAM{} model using $\poly(\card{E},L)$ processors in  $\polylog(\card{E},L)$ expected time.
\end{theorem}

	\section{Hierarchical Exploration Time}
\label{sec:exploration}
The last element needed to obtain an \RNC{} algorithm for the random sampling of arborescences is to show that for the decomposition achieved above, for polynomially large $L$, simulating the first $L$ edges which jump across each node $\sS \in  \T$ will contain all first-visit edges to each vertex with high probability. We do this by bounding the number of edge traversals between children of each node before that node is covered. We build on techniques developed by \textcite{BPS18}.

We start by bounding the number of times a given vertex is visited before cover time. For any $v, s, t \in V$, let $H_v(s, t)$ denote the expected number of times a random walk starting at $s$ visits $v$ before reaching $t$; in the case that $s = t$, $H_v(s,s)$ denotes the number of times a walk starting at $s$ reaches $v$ before returning to $s$.

It is easy to see that $H_v$ satisfies a triangle inequality, namely:

\begin{proposition}\label{prop:triangle}
On any directed graph $G = (V, E)$, for any $v, s, t, u , \in V$: $H_v(s, t) \leq H_v(s, u) + H_v(u, t)$.
\end{proposition}

Additionally, as the stationary distribution of vertices on an Eulerian graph is proportional to the degree of a vertex, it follows that:

\begin{proposition}\label{prop:timeBetweenReturns}
On any directed Eulerian graph $G = (V, E)$, for any $v,s\in V$: $H_v(s, s) = \frac{1}{\pi(s)}\pi(v) = \frac{\deg(v)}{\deg(s)}$
\end{proposition}

\begin{proof}
Consider the Schur complement $G_S$ where $S = \{v, s\}$. Since $G$ is Eulerian, so is $G_{S}$ and the degrees of $v, s$ are the same in $G$ and $G_S$. Letting the weights of edges in $G_S$ be $w(v, s)$, $w(s, v)$, $w(v, v)$, $w(s, s)$, we must have $w(s, v) = w(v, s)$ for $G_S$ to be Eulerian. As random walks on $G_S$ are distributed like random walks on $G$ whose transcripts have been restricted to only contain vertices in $S$,  the values of $H_v(s, s), H_v(v, s)$ are the same for a walk on $G$ as for one on $G_S$. Considering one-step transitions, we can construct the following system of equations:
\begin{align*}
H_v(s, s) &= \frac{w(s, v)}{\deg(s)} \left ( H_v(v, s) +1\right) + \frac{w(s,s)}{\deg(s)}0,\\
H_v(v, s) &= \frac{w(v, v)}{\deg(v)} \left ( H_v(v, s) +1\right) + \frac{w(v,s)}{\deg(v)}0.
\end{align*}
Solving yields $H_v(v, s) = \frac{w(v, v)}{w(v, s)}$ and $H_v(s, s) =  \frac{w(s, v)}{\deg(s)} \left (\frac{w(v, v)}{w(v, s)} +1\right) = \frac{w(s, v)}{\deg(s)} \frac{\deg(v)}{w(v, s)}=\frac{\deg(v)}{\deg(s)}$.

\end{proof}


For any two vertices which are connected by an edge, we have:

\begin{lemma} \label{lem:edgeExploration}
On any directed Eulerian graph $G = (V, E)$, for any $v, s, t\in V$ such that $e = (s, t) \in E$, $H_v(s, t) \leq \frac{\deg(v)}{w(s, t)}$
\end{lemma}

\begin{proof}
Note that every time the walk is currently at $s$, there is a $\frac{w(s, t)}{\deg(s)}$ chance of moving to $t$, thus, in expectation, the expected number of times the walk returns to $s$ before reaching $t$ is at most $\frac{\deg(s)}{w(s,t)}$. Between each return to $s$, the expected number of visits to $v$ is $\frac{\deg(v)}{\deg(s)}$ by \cref{prop:timeBetweenReturns}. Thus, we conclude that:
\[H_v(s, t) \leq \frac{\deg(s)}{w(s,t)}\frac{\deg(v)}{\deg(s)} = \frac{\deg(v)}{w(s,t)}. \qedhere \]
\end{proof}

On graphs which are strongly connected by edges of some weight $w_c$, this observation leads to the conclusion that: 

\begin{lemma} \label{lem:connectedExploration}
On an Eulerian directed graph $G = (V, E)$ strongly connected by edges of weight $w_c$, for any $s, t \in V$, $H_s(s, t) \leq \frac{\deg(s) n}{w_c}$ 
\end{lemma}

\begin{proof}
By connectivity assumptions, there exists a sequence of at most $n$ vertices $s=v_0, v_1, \dots, v_k=t$ which form a path connecting $s$ and $t$ such that $w(v_{i-1}, v_i) \geq w_c$ for all $i \in [k]$. Then, combining \cref{prop:triangle} with \cref{lem:edgeExploration},
\[H_s(s, t) \leq \sum_{i = 1}^{k}H_s(v_{i-1}, v_i) \leq \sum_{i = 1}^{k}\frac{\deg(s)}{w(v_{i-1}, v_i)} \leq \frac{n \deg(s) }{w_c}. \qedhere\]
\end{proof}

This trivially bounds the expected number of returns to a vertex $s$ before a graph $G$ is covered by a random walk by $\frac{ n^2 \deg(s)}{w_c}$. 

\begin{remark}
While this bound is sufficient for our purposes, using Matthew's trick, this bound can be further tightened to $\frac{n\log n \deg(s)}{w_c}$ \cite{MAT88}.
\end{remark}

\begin{lemma}\label{lem:clusterExploration}
For any node $\mathcal{S} \in \mathcal{T}$ in the decomposition obtained by \textup{\cref{alg:decomposition}}, the expected number of edges traversed between children of $\mathcal{S}$ before every vertex in $\mathcal{S}$ has been visited is at most $n^2m^2$.
\end{lemma}
\begin{proof}
For each edge $e = (u, v)$ which jumps between children of $\sS$, the expected number of times edge $e$ is traversed is $\frac{w_e}{\deg(u)}$ times the expected number of times vertex $u$ is reached by a random walk before all vertices in  $\sS$ are reached. 
To bound the number of times $u$ is reached before $\sS$ is covered, consider a random walk on the Schur complement $G_{\sS}$. The expected number of times $u$ is reached in a random walk on $G$ before $\sS$ is covered is the same as the expected number of times $u$ is reached by a random walk on $G_{\sS}$ before $G_{\sS}$ is covered.

By \cref{prop:eulerianSchure}, $G_{\sS}$ will be Eulerian and the degrees of all vertices in $\sS$ will be the same in $G$ and $G_\sS$. By \cref{lem:connectedness}, $\sS$ is strongly connected by edges of weight $\frac{w_{\text{max}}(\sS)}{m}$, and so $G_{\sS}$ is as well. Then, by \cref{lem:connectedExploration}, the expected number of times $u$ is hit before ${G_{\sS}}$ is covered is at most by $\frac{\deg(u)n^2m}{w_{\text{max}}}$. Thus, the expected number of times edge $e$ is traversed before $\sS$ is covered is at most
\[\frac{\deg(u)n^2m}{w_{\text{max}}}\frac{w_e}{\deg(u)} \leq n^2 m\] since $w_c \geq \frac{w_e}{m}$. As there are at most $m$ jumping edges of $\mathcal{S}$, in expectation,  at most $n^2m^2$ edge traversals between children of $\mathcal{S}$ will occur before every vertex in $\sS$ has been reached.
\end{proof}

Thus, for large enough $L = \poly(n, m)$, by Markov's inequality, with high probability every node $\sS \in \T$ is covered by the time $L$ edges have jumped across $\sS$. This means the transcript returned by \texttt{AllEdges} will contain all first visit edges with high probability. As this transcript is polynomial in length, all first visit edges and the corresponding arborescence they form can be extracted and returned in \NC{}. 
	\section{Discussion and Open Problems}
\label{sec:discussion}	

We showed how to sample arborescences, and as a special case spanning trees, from a given weighted graph using an \RNC{} algorithm. While this is a step in resolving the disparity between parallel sampling and parallel counting algorithms, more investigation is needed. In particular, for the list of problems with determinant-based counting in \cref{sec:intro}, designing \RNC{} sampling algorithms remains open.

One of these problems in particular, namely sampling Eulerian tours from an Eulerian digraph, is intimately connected to sampling arborescences, due to the BEST theorem \cite{BE87, ST41}. However the reduction, while polynomial-time implementable, is not known to be in \NC{}.
\begin{question}
		Is there an algorithm for sampling Eulerian tours uniformly at random in Eulerian digraphs?
\end{question}

The important tool our result relied on was the Aldous-Broder algorithm. So it is natural to ask whether there are generalizations of the Aldous-Broder result to settings beyond spanning trees and arborescences. In particular the bases of a \emph{regular matroid} are a proper generalization of spanning trees in a graph, and they have a decomposition in terms of graphic, co-graphic, and some special constant-sized matroids \cite{seymour1980decomposition}.
\begin{question}
	Can we sample from a regular matroid, or equivalently from a volume-based distribution defined by totally unimodular vectors in \RNC{}?
\end{question}

Yet another direction for generalization are higher-dimensional equivalents of spanning trees and arborescences. For example, \textcite{GP14} provided a generalization of the algorithm of \textcite{Wil96} for sampling arborescences from graphs to hypergraphs. While Wilson's algorithm is different from that of Aldous-Broder, it is closely related. Can these generalizations be efficiently parallelized?

	\PrintBibliography
\end{document}